\documentclass[a4paper]{article}

\usepackage[cmex10]{amsmath}
\usepackage{amssymb}
\usepackage{amsthm}
\usepackage{enumerate}
\usepackage{amscd}

\theoremstyle{plain}
\newtheorem{theorem}{Theorem}
\newtheorem{lemma}[theorem]{Lemma}
\newtheorem{proposition}[theorem]{Proposition}

\newtheorem{definition}[theorem]{Definition}

\theoremstyle{definition}
\newtheorem{remark}[theorem]{Remark}

\usepackage{url}

\newcommand{\beq}{\begin{equation*}}
\newcommand{\eeq}{\end{equation*}}

\DeclareMathOperator{\dmin}{d_{min}}
\DeclareMathOperator{\Supp}{Supp}

\newcommand{\Fq}{{\mathbb{F}_q}}

\newcommand{\F}{\mathbb{F}}
\newcommand{\ie}{\emph{i.e. }}

\newcommand{\moins}{\setminus}

\newcommand{\deux}[1][2]{^{\langle #1\rangle}}
\newcommand{\supco}{\emph{support condition}}

\def\epsilon{\varepsilon}
\def\subset{\subseteq}

\begin{document}

\title{An upper bound of Singleton type for componentwise products of linear codes}

\author{Hugues Randriambololona}

\maketitle

\begin{abstract}
We give an upper bound that relates the minimum weight of a nonzero componentwise product of codewords
from some given number of linear codes, with the dimensions of these codes.
Its shape is a direct generalization of the classical Singleton bound.
\end{abstract}

\section{Introduction}
\label{intro}

Let $q$ be a prime power,
and $\Fq$ the field with $q$ elements.
For any integer $n\geq1$, let $*$ denote componentwise multiplication
in the vector space $(\Fq)^n$, so
\beq
(x_1,\dots,x_n)*(y_1,\dots,y_n)=(x_1y_1,\dots,x_ny_n).
\eeq
If $C_1,\dots,C_t\subset(\Fq)^n$ are linear codes of the same length~$n$,
let
\beq
C_1*\cdots*C_t=\sum_{c_i\in C_i}\Fq\cdot c_1*\cdots*c_t\;\subset\,(\Fq)^n
\eeq
be the linear code spanned by the componentwise products of their codewords.
(In \cite{agis} this was denoted $\langle C_1*\cdots*C_t \rangle$
with brackets meant to emphasize that the linear span is taken.
Here we will keep notation lighter. All codes in this text will be linear.)

Also define the square of a linear code $C$ as
the linear code $C\deux=C*C$, and likewise for its higher powers $C\deux[t]$.

Basic properties of these operations, as well as a geometric interpretation,
will be found in \cite{AGCT}.

\medskip

Bounds on the possible joint parameters of $C$ and $C\deux[t]$, 
or more generally on that of some $C_i$ and their product $C_1*\cdots*C_t$,
have attracted attention recently for various reasons:
\begin{itemize}
\item they determine the performance of bilinear multiplication algorithms,
in particular against random or adversarial errors,
or against eavesdropping; this is useful either
in questions of algebraic complexity \cite{LW}\cite{ChCh+},
or in the study of secure multi-party computation systems \cite{CCCX}
\item since $C_1*\cdots*C_t$ captures possibly hidden algebraic relations
between subcodes $C_i$ of a larger code (given by an apparently
random generator matrix), they're at the heart of attacks \cite{CGGOT}
against McEliece type cryptosystems
\item following \cite{ZO}, the existence of asymptotically good
binary linear codes with asymptotically good squares is the key ingredient
in an improvement of the Cr\'epeau-Kilian \cite{CK} oblivious transfer protocol
over a noisy channel;
solving this problem was the main motivation for \cite{agis}
\item last, this $*$ operation is also of use in the understanding
of algebraic decoding algorithms through the notion of error-locating
pairs \cite{P}.
\end{itemize}

While it is possible to give bounds involving subtler parameters,
such as the dual distance (see Lemma \ref{dim+dperp} below for an
elementary example, or \cite{M} for a more elaborate result),
here we want to deal with ``clean'' bounds involving only the
dimensions of the $C_i$ and the minimum distance of $C_1*\cdots*C_t$.
In particular we will study the following
generalizations (introduced in \cite{agis}) of the
fundamental functions of block coding theory:
\beq
a_q\deux[t](n,d)=\max\{k\geq0\,|\,\exists C\subset(\Fq)^n,\,\dim(C)=k,\,\dmin(C\deux[t])\geq d\}
\eeq
and
\beq
\alpha_q\deux[t](\delta)=\limsup_{n\to\infty}\frac{a_q\deux[t](n,\lfloor\delta n\rfloor)}{n}.
\eeq
In fact, for $t\geq2$ we have the easy inequalities $\dim(C\deux[t])\geq\dim(C\deux[t-1])$
and $\dmin(C\deux[t])\leq\dmin(C\deux[t-1])$ (see \cite{agis}, Prop.~11), from which
one deduces
\beq
a_q\deux[t](n,d)\leq a_q\deux[t-1](n,d)\leq\dots\leq a_q(n,d)
\eeq
\beq
\alpha_q\deux[t](\delta)\leq\alpha_q\deux[t-1](\delta)\leq\dots\leq\alpha_q(\delta)
\eeq
where $a_q(n,d)$, $\alpha_q(\delta)$, are the usual,
much-studied fundamental functions; hence all the upper bounds known
on these functions apply. Here we will get a new, stronger bound, by working
directly on the generalized functions.

\medskip

The paper is organized as follows. In Section~\ref{main} we state and prove our main result, the product
Singleton bound, in full generality.
In Section~\ref{alternative} we propose an alternative proof that works only in a special case,
and moreover leads to a slightly weaker result;
but it uses entirely different methods that could be of independent interest.
Then in Section~\ref{appl} we derive our new upper bound on the fundamental functions;
in particular for $d\leq t$ we get the exact value of $a_q\deux[t](n,d)$.

\medskip

\textbf{Notations.}
We let $[n]=\{1,\dots,n\}$ be the standard set with $n$ elements.
Given a subset $I\subset[n]$, we let $\pi_I:(\Fq)^n\to(\Fq)^I$ be the natural projection.

\section{The product Singleton bound}
\label{main}

Here we state our main result, which for $t=1$ reduces to the (linear version of the) classical
Singleton bound.
For this we introduce a mild technical condition (which will be discussed further in Remark~\ref{laremarque} below).

\begin{definition}
\label{ladef}
Let $t\geq3$ be an integer and let $C_1,\dots,C_t\subset(\F_q)^n$ be linear codes
of the same length $n$. We say these $C_i$ satisfy the {\supco} if, for each coordinate $j\in[n]$,
either $j$ is in the support of \emph{all} the $C_i$, or it is in the support of \emph{at most one}
of them.
\end{definition}

\begin{theorem}
\label{letheo}
Let $t\geq1$ be an integer and let $C_1,\dots,C_t\subset(\F_q)^n$ be linear codes
of dimension $k_1,\dots,k_t$ respectively, and of the same length $n$.
Suppose $C_1*\cdots*C_t\neq0$, and if $t\geq3$ suppose they satisfy the {\supco}.
Then one can find codewords $c_i\in C_i$ such that their product $c_1*\cdots*c_t$ has weight
\beq
1\leq w(c_1*\cdots*c_t)\leq\max(t-1,\,n+t-(k_1+\cdots+k_t)\,).
\eeq
As a consequence, $\dmin(C_1*\cdots*C_t)\leq\max(t-1,\,n+t-(k_1+\cdots+k_t)\,)$.
\end{theorem}

This upper bound is tight. For example it is attained when the $C_i$ are Reed-Solomon codes,
with $k_1+\cdots+k_t\leq n$.

Also when $k_1+\cdots+k_t>n$, the upper bound $t-1$ can be attained.
For an example with $t=3$ consider the code $C$ with generator matrix
$G=\left(\begin{array}{cccc}1&1&0&0\\0&0&1&1\end{array}\right)$.
Then $tk=6>n=4$, and $C\deux[3]=C$ has $\dmin=t-1=2$.

Note that the existence of the $c_i$ is stronger than the bound
on the minimum distance alone: indeed, in general $\dmin(C_1*\cdots*C_t)$ need not
be attained by a codeword $z$ in specific product form $z=c_1*\cdots*c_t$
(one might need \emph{a linear combination} of such codewords).
However, what makes the proof difficult is that, while we want the intersection of
the supports of the $c_i$
to be small, at the same time we need to ensure it remains nonempty.

\begin{remark}
\label{laremarque}
Here we want to make a few comments about the {\supco}:
\begin{enumerate}[(a)]
\item
Although this {\supco} for $t\geq3$ might seem a little bit restrictive, in fact it is satisfied in many
important situations. For instance,
it is satisfied when $C_1,\dots,C_t$ all have full support, or when
$C_1=\dots=C_t=C$ are all equal to the same code $C$
(not necessarily of full support).
\item
However, the conclusion in Theorem~\ref{letheo} can fail if one drops the {\supco}.
For example, the codes $C_1,C_2,C_3$ with generator matrices
$G_1=G_2=\left(\begin{array}{cccc}1&1&1&0\\0&0&0&1\end{array}\right)$
and $G_3=\left(\begin{array}{cccc}1&1&1&0\end{array}\right)$
have $k_1+k_2+k_3=5>n=4$, but $\dmin(C_1*C_2*C_3)=3$.
\item
Here we assume that we have a proof of Theorem~\ref{letheo} in the particular case
where all $C_i$ have full support.

First, this allows us to deduce the following \emph{unconditional} variant:

Let $C_1,\dots,C_t\subset(\F_q)^n$ be \emph{any} linear codes
of the same length $n$. Suppose $I=\bigcap_i\Supp(C_i)\neq\emptyset$,
and let $\overline{n}=|I|$, $\overline{k}_i=\dim\pi_I(C_i)$.
Then one can find codewords $c_i\in C_i$ such that their product has weight
\beq
1\leq w(c_1*\cdots*c_t)\leq\max(t-1,\,\overline{n}+t-(\overline{k}_1+\cdots+\overline{k}_t)\,).
\eeq

Indeed, the codes $\pi_I(C_i)$ all have full support in $I$,
so by our assumption one can find codewords $\pi_I(c_i)\in\pi_I(C_i)$
satisfying the estimates.
Then just observe that $w(c_1*\cdots*c_t)=w(\pi_I(c_1)*\cdots*\pi_I(c_t))$.

Now we claim that, in turn, this implies the full statement of Theorem~\ref{letheo}. Indeed suppose
the $C_i$ satisfy the {\supco} if $t\geq3$. Write $\Supp(C_i)=I\cup J_i$.
The $J_i$ are disjoint: this is obvious if $t\leq2$, and if $t\geq3$
this is precisely the meaning of the {\supco}.
Then we have $\overline{n}\leq n-(|J_1|+\cdots+|J_t|)$,
while $\overline{k}_i\geq k_i-|J_i|$, hence
\beq
\overline{n}+t-(\overline{k}_1+\cdots+\overline{k}_t)\leq n+t-(k_1+\cdots+k_t)
\eeq
which finishes the proof.
\end{enumerate}
\end{remark}

Thanks to the equivalence of the statements in the last remark, we see that to prove
Theorem~\ref{letheo}, it suffices to do so under the additional assumption that all the codes
have full support. The key step in the proof will be the following lemma, which treats
the case of ``high dimension''.

\begin{lemma}
\label{lelemme}
Let $C_1,\dots,C_t\subset(\F_q)^n$ be linear codes
of dimension $k_1,\dots,k_t$ respectively, and of the same length $n$.
Suppose these codes all have full support, and
\beq
k_1+\cdots+k_t>n.
\eeq
Then one can find codewords $c_i\in C_i$ such that
\beq
1\leq w(c_1*\cdots*c_t)\leq t-1.
\eeq
\end{lemma}
\begin{proof}
If $H$ is a matrix with $n$ columns, we say that a subset $A\subset[n]$ is dependent
(resp. independent, maximal independent) in $H$ if, in the set of columns of $H$, those indexed by $A$
form a linearly dependent (resp. independent, maximal independent) family.

Now let $H_i$ be a parity-check matrix for $C_i$.
We claim that we can find subsets $A_1,\dots,A_t\subset[n]$, and an element $j_1\in[n]$,
such that:
\begin{enumerate}[(1)]
\item\label{intvide} $A_1\cap\dots\cap A_t=\emptyset$
\item\label{A1indep} $A_1$ is independent in $H_1$
\item\label{Aimaxindep} $A_i$ is maximal independent in $H_i$ for $i\geq2$
\item\label{j1} $A_1\cup\{j_1\}$ is dependent in $H_1$, and $A_2\cup\{j_1\}$ is dependent in $H_2$.
\end{enumerate}

These are constructed as follows. First, for all $i$, choose any $B_i\subset[n]$ maximal
independent in $H_i$, and let $I=B_1\cap\dots\cap B_t$ be their intersection.
Then $|B_1|+\cdots+|B_t|=tn-(k_1+\cdots+k_t)<(t-1)n$,
so there exists $j_1\in[n]$ that belongs to at most $t-2$ of the sets $B_i$. Say $j_1\not\in B_1$
and $j_1\not\in B_2$.

Suppose $(B_1\moins I)\cup\{j_1\}$ is independent
in $H_1$. Then $I$ is nonempty (otherwise $B_1$ would not be maximal), and
by the basis exchange property
from elementary linear algebra, one can find $j\in I$ such that $(B_1\moins\{j\})\cup\{j_1\}$
is maximal independent in $H_1$. Then we replace $B_1$ with $(B_1\moins\{j\})\cup\{j_1\}$,
which replaces $I$ with $I\moins\{j\}$.

We repeat this procedure until, obviously, it must stop,
which means $(B_1\moins I)\cup\{j_1\}$
is dependent in $H_1$. Then we set $A_1=B_1\moins I$, and $A_i=B_i$ for $i\geq2$.




\medskip

Now that this is done,
by \eqref{A1indep} and \eqref{j1} there is $c_1\in C_1$ with
\beq
\{j_1\}\subset\Supp(c_1)\subset A_1\cup\{j_1\},
\eeq
and likewise by \eqref{Aimaxindep} and \eqref{j1} there is $c_2\in C_2$ with
\beq
\{j_1\}\subset\Supp(c_2)\subset A_2\cup\{j_1\},
\eeq
hence
\beq
\{j_1\}\subset\Supp(c_1*c_2)\subset (A_1\cap A_2)\cup\{j_1\}.
\eeq
This means we have established the step $s=2$ in the following induction procedure:

Suppose for some $s\leq t$
we have found indices $j_1,\dots,j_{s-1}\in[n]$ (not necessarily distinct) and codewords
$c_1\in C_1,\dots,c_{s}\in C_{s}$ (after possibly renumbering), such that:
\begin{enumerate}[(1)]
\setcounter{enumi}{4}
\item\label{decsupport}
$\displaystyle\;\{j_{s-1}\}\subset\Supp(c_1*\cdots*c_{s})\subset (A_1\cap\dots\cap A_{s})\cup\{j_1,\dots,j_{s-1}\}.$
\end{enumerate}

If $s=t$, the proof is finished thanks to condition~\eqref{intvide}. 
So we suppose $s<t$, and we will show how to pass from $s$ to $s+1$ in the induction.

By~\eqref{decsupport} we can write
\beq
\Supp(c_1*\cdots*c_{s})=S\cup T
\eeq
with
\beq
S\subset A_1\cap\dots\cap A_{s}
\eeq
and
\beq
\{j_{s-1}\}\subset T\subset\{j_1,\dots,j_{s-1}\}.
\eeq
We distinguish two cases.

First, suppose $S=\emptyset$. Set $j_s=j_{s-1}$. Then we can find $c_{s+1}\in C_{s+1}$ nonzero at $j_s$
(because $C_{s+1}$ has full support), and we're done.

Otherwise, suppose $S\neq\emptyset$, so there is $j_{s}\in S$. By~\eqref{intvide},
there is $i>s$ such that $j_{s}\not\in A_i$. Say this is $i=s+1$.
Then, by~\eqref{Aimaxindep}, one can find $c_{s+1}\in C_{s+1}$ such that
\beq
\{j_{s}\}\subset\Supp(c_{s+1})\subset A_{s+1}\cup\{j_{s}\},
\eeq
from which it follows
\beq
\{j_{s}\}\subset\Supp(c_1*\cdots*c_{s+1})\subset (A_1\cap\dots\cap A_{s+1})\cup\{j_1,\dots,j_{s}\}.
\eeq
The proof is complete.
\end{proof}

\begin{proof}[End of the proof of Theorem~\ref{letheo}]
Thanks to Remark~\ref{laremarque}(c) we can assume all $C_i$ have full support.
Also we assume $k_1+\cdots+k_t\leq n$, otherwise it suffices to apply Lemma~\ref{lelemme}.

We conclude with the same puncturing argument as in one of the proofs of the classical Singleton bound:
let $\pi$ denote projection on the first $(k_1+\cdots+k_t)-1$ coordinates.
We distinguish two cases.

First, suppose $\dim(\pi(C_i))=\dim(C_i)=k_i$ for all $i$.
Then we can apply Lemma~\ref{lelemme} and we get $\pi(c_i)\in\pi(C_i)$ such that
$1\leq w(\pi(c_1)*\cdots*\pi(c_t))\leq t-1$.
Lifting back we find $1\leq w(c_1*\cdots*c_t)\leq n+t-(k_1+\cdots+k_t)$,
which finishes the proof.

Otherwise, if this fails say for $i=1$,
there is $c_1\in C_1$ nonzero in $\ker(\pi)$, so $w(c_1)\leq n+1-(k_1+\cdots+k_t)$.
Fix a coordinate $j\in\Supp(c_1)$ and for each $i\geq2$ take $c_i\in C_i$ nonzero at $j$
(which is possible since $C_i$ has full support).
Then $c_1*\cdots*c_t$ is nonzero with weight
$w(c_1*\cdots*c_t)\leq w(c_1)\leq n+1-(k_1+\cdots+k_t)\leq n+t-(k_1+\cdots+k_t)$,
as needed.
\end{proof}

Observe that our proof of Theorem~\ref{letheo} is constructive:
$c_1,\dots,c_t$ can be effectively computed from
given parity-check matrices of the codes.

\section{An alternative proof for $t=2$}
\label{alternative}

Consider the following statement, that is easily seen to be a special case of Theorem~\ref{letheo}.

\begin{proposition}
\label{laprop}
Let $C,C'\subset(\F_q)^n$ be linear codes of dimension $k,k'$ respectively,
and of the same length $n$. Then their product $C*C'$ has minimum distance
\beq
\dmin(C*C')\leq\max(1,n-k-k'+2).
\eeq
\end{proposition}

Compared with Theorem~\ref{letheo}, an obvious restriction is that we consider the product
of only $t=2$ codes. But Proposition~\ref{laprop} is also less precise: given $C*C'\neq 0$,
it says there is a nonzero codeword $z$ of weight at most $\max(1,n-k-k'+2)$, but it does not
give any information on it; while from Theorem~\ref{letheo}, we know it can be taken
in elementary product form $z=c*c'$ (and moreover it can be effectively computed).

However Proposition~\ref{laprop} can be proved using entirely different methods.
For this we will need two lemmas.

\begin{lemma}
\label{dim+dperp}
Let $C_1,C_2\subset(\Fq)^n$ be two linear codes.
Suppose both $C_1,C_2$ have dual minimum distance at least $2$,
\ie full support. Then:
\beq
\dim(C_1*C_2)\geq\min(n,\,\dim(C_1)+\dmin(C_2^\perp)-2).
\eeq
\end{lemma}
\begin{proof}
Set $k_1=\dim(C_1)$, $d_2^\perp=\dmin(C_2^\perp)$,
and $m=\min(n,k_1+d_2^\perp-2)$.
Then $m-k_1+1\leq d_2^\perp-1$, so any $m-k_1+1$ columns of $C_2$ are
linearly independent, in particular:

\textbf{Fact.} For any set of indices $J\subset[n]$ of size $|J|=m-k_1$,
and for any $j_0\not\in J$,
there is a codeword $y\in C_2$ with $y_{j_0}=1$
and $y_{j}=0$ for $j\in J$.

Now (after possibly permuting coordinates) put $C_1$ in systematic form,
with generator matrix $G_1=(I_{k_1}|X)$. To show $\dim(C_1*C_2)\geq m$,
we will find, for each $i\in[m]$, a codeword $z\in C_1*C_2$
with $z_i\neq 0$ and $z_j=0$ for $j\in[m]\moins\{i\}$.
We distinguish two cases.

First, suppose $i\in[k_1]$. Let $x$ be the $i$-th row of $G_1$,
and let $y$ be given by the Fact with $j_0=i$
and $J=[m]\moins[k_1]$. 
Then we can set $z=x*y$.

Otherwise, suppose $i\in[m]\moins[k_1]$. Since $C_1$ has full support, there is
a row of $G_1$ that is nonzero at $i$. Say this is the $i'$-th row,
and denote it by $x$.
Now let $y$ be given by the Fact with $j_0=i$
and $J=\{i'\}\cup([m]\moins([k_1]\cup\{i\}))$.
Then again $z=x*y$ does the job.
\end{proof}

\begin{lemma}
\label{adj}
For any two linear codes $C,C'\subset(\Fq)^n$ we have
\beq
C\;\;\perp\;\; C'*(C*C')^\perp.
\eeq
\end{lemma}
\begin{proof}
Let $\tau:(\Fq)^n\to\Fq$ be the ``trace'' linear map,
$\tau(x_1,\dots,x_n)=x_1+\cdots+x_n$.
Note that the canonical scalar product $\langle\cdot|\cdot\rangle$ on $(\Fq)^n$
can be written as
$\langle c|c'\rangle=\tau(c*c')$.
Now for any $c\in C$, $c'\in C'$, and $x\in(C*C')^\perp$, we have
\beq
\langle c|c'*x\rangle=\tau(c*(c'*x))=\tau((c*c')*x)=\langle c*c'|x\rangle=0
\eeq
and we conclude by passing to the linear span.
\end{proof}

\medskip

We can now proceed. In what follows let
$\widetilde{d}=\dmin(C*C')$.

\medskip

\begin{proof}[Proof of Proposition~\ref{laprop}]
It suffices to treat the case where $C$ and $C'$ both have full
support. For then, to deduce the case of general $C$ and $C'$, just project
on the intersection of their supports: this leaves $\widetilde{d}$ unchanged,
while $n-k-k'$ can only decrease (this is the very same argument as in Remark~\ref{laremarque}(c)).

Also suppose $\widetilde{d}\geq2$, otherwise there is nothing to prove.

That $C$ has full support implies that for any $c'\in C'$ of minimum weight $d'=\dmin(C')$,
there is a $c\in C$ whose support intersects that of $c'$ non-trivially, meaning $c*c'\neq0$:
this implies $\widetilde{d}\leq d'$, hence by the classical
Singleton bound
\beq
k'\leq n-\widetilde{d}+1.
\eeq
Lemma~\ref{dim+dperp} applied with $C_1=C'$ and $C_2=(C*C')^\perp$
then gives
\beq
\dim(C'*(C*C')^\perp)\geq k'+\widetilde{d}-2,
\eeq
and by Lemma~\ref{adj} we conclude
\beq
k\leq n-\dim(C'*(C*C')^\perp)\leq n-k'-\widetilde{d}+2
\eeq
as needed.
\end{proof}

In the author's opinion, Lemmas~\ref{dim+dperp} and~\ref{adj} are very natural and have interest
on their own.
But the way they combine to give this concise but not-so-intuitive proof of Proposition~\ref{laprop}
is quite intriguing.

\section{Upper bound on the generalized fundamental functions}
\label{appl}

An important consequence of Theorem~\ref{letheo} is the following:

\begin{theorem}
\label{lecoro}
We have
$a_q\deux[t](n,d)=\left\lfloor\frac{n}{d}\right\rfloor$
for $1\leq d\leq t$, and
\beq
a_q\deux[t](n,d)\leq\left\lfloor\frac{n-d}{t}\right\rfloor+1\qquad\text{for $t<d\leq n$.}
\eeq
Likewise, $\alpha_q\deux[t](0)=1$, and
\beq
\alpha_q\deux[t](\delta)\leq\frac{1-\delta}{t}\qquad\text{for $0<\delta\leq 1$.}
\eeq
\end{theorem}
\begin{proof}
Suppose first $d\geq t$. If $C$ has parameters $[n,k]$ and $\dmin(C\deux[t])\geq d$, Theorem~\ref{letheo} applied
with all $C_i=C$ gives $d\leq n-(k-1)t$, from which the bound
$a_q\deux[t](n,d)\leq\left\lfloor\frac{n-d}{t}\right\rfloor+1$ follows.

In particular, on the ``diagonal'' $d=t$ we find
$a_q\deux[t](n,t)\leq\left\lfloor\frac{n-t}{t}\right\rfloor+1=\left\lfloor\frac{n}{t}\right\rfloor$
for all $t$.
 
Then for $d<t$, we deduce
$a_q\deux[t](n,d)\leq a_q\deux[d](n,d)\leq\left\lfloor\frac{n}{d}\right\rfloor$.

To show that this upper bound is in fact an equality for $d\leq t$, partition the set $[n]$
of coordinates
into $\left\lfloor\frac{n}{d}\right\rfloor$ subsets of size $d$ or $d+1$, and consider
the code $C$ spanned by their characteristic vectors (observe $C\deux[t]=C$).

This done, letting $n\to\infty$ and normalizing then gives the estimate on
$\alpha_q\deux[t](\delta)$. (For the special value $\alpha_q\deux[t](0)=1$,
we used $a_q\deux[t](n,1)=n$.)
\end{proof}

Note in particular that for $t\geq2$, the function $\alpha_q\deux[t](\delta)$
is not continuous at $\delta=0$, in striking contrast with the ``usual''
function $\alpha_q(\delta)$. Perhaps one could modify the definition of $\alpha_q\deux[t]$
to remove this discontinuity. Nevertheless it remains $\limsup\alpha_q\deux[t](\delta)\leq\frac{1}{t}<1$
as $\delta\to0$.
Thus for small $\delta$ our bound clearly improves
on the inequality $\alpha_q\deux[t](\delta)\leq\alpha_q(\delta)$,
and in fact one can show it is so for all $\delta<1-\epsilon(q)$,
with $\epsilon(q)\to 0$ as $q\to\infty$.

\medskip

Conversely
it is interesting to compare the upper bound in Theorem~\ref{lecoro} with
known lower bounds. From algebraic-geometry codes one easily gets
(see \cite{agis} for more details)
\beq
\alpha_q\deux[t](\delta)\geq\frac{1-\delta}{t}-\frac{1}{A(q)}
\eeq
where $A(q)$ is the Ihara constant. When $q\to\infty$,
the two bounds match.
On the other hand, for $q$ small, the two bounds remain far apart.
For $t=2$, even with the improved lower bound of \cite{agis}, namely
\beq
\alpha_q\deux(\delta)\geq\frac{1}{s+1}\left(\frac{1}{1+q^s}-\frac{1}{A(q^{2s+1})}\right)-\frac{2s+1}{1+q^s}\,\delta
\eeq
(for any $s\geq0$),
there remains much room for progress.
For instance, for $q=2$, the best we get ($s=4$) is
\beq
0.001872-0.5294\,\delta\;\leq\; \alpha_2\deux(\delta)\;\leq\; 0.5-0.5\,\delta.
\eeq
Still for $q$ small, the situation for $t\geq3$ is even worse:
no nontrivial lower bound on $\alpha_q\deux[t]$ is known then!

\section*{Acknowledgment}

The author is indebted to the Associate Editor Dr.~Navin Kashyap for his
proof of the case~$t=2$ of what is now Lemma~\ref{lelemme}.
This was the key starting point in a series of improvements that led
from Proposition~\ref{laprop} (which the author had beforehand)
to the now much more general Theorem~\ref{letheo}.

\end{document}